\documentclass[%
 reprint,
superscriptaddress,
]{revtex4-2}
\usepackage[utf8]{inputenc}
\usepackage{graphicx}
\usepackage{dcolumn}
\usepackage{dsfont}
\usepackage{amssymb}
\usepackage{MnSymbol}
\usepackage{amsfonts}
\usepackage{mathtools}
\usepackage{amsmath}
\usepackage{amsthm}
\usepackage{xcolor}
\usepackage{hyperref}
\usepackage{comment}
\usepackage{physics}

\newcommand\mc[1]{\mathcal{#1}}
\newcommand{\cA}{{\mc{A}}}

\newcommand{\cB}{{\mc{B}}}

\newcommand{\cD}
{{\mc{D}}}

\newcommand{\cH}{{\mc{H}}}
\newcommand{\cG}{{\mc{G}}}

\newcommand{\cM}{{\mc{M}}}

\newcommand{\cP}{{\mc{P}}}

\newcommand{\aP}{{\widetilde{\mc{P}}}}

\newcommand{\cS}{{\mc{S}}}

\newcommand{\cX}{{\mc{X}}}
\newcommand{\tX}{{\widetilde{\cX}}}

\newcommand{\ta}{{\widetilde{a}}}
\newcommand{\tb}{{\widetilde{b}}}

\newcommand{\tx}{{\widetilde{x}}}

\newcommand{\tA}{{\widetilde{A}}}

\newcommand{\tB}{{\widetilde{B}}}

\newcommand{\tM}{{\widetilde{\cM}}}
\newcommand{\tO}{{\widetilde{O}}}

\newcommand{\oS}{{\overline{\cS}}}

\newcommand{\cU}{{\mc{U}}}

\newcommand{\R}{{\mathbb R}}
\newcommand{\C}{{\mathbb C}}
\newcommand{\CP}{{\mathbb{CP}}}

\newcommand{\Cl}{{\mc{C}}}
\newcommand{\one}{{\mathds{1}}}

\newcommand{\zz}{{\mathbb{Z}}}
\newcommand{\ra}{{\rightarrow}}

\newcommand{\PI}{{\cS}}

\newtheorem{definition}{Definition}
\newtheorem{lemma}{Lemma}
\newtheorem{corollary}{Corollary}
\newtheorem{proposition}{Proposition}
\newtheorem{theorem}{Theorem}

\begin{document}
\preprint{APS/123-QED}

\title{Maximal complementarity in the $n$-qubit Pauli group}
\author{Markus Frembs}
\email{markus.frembs@itp.uni-hannover.de}
\thanks{(first author)}
\affiliation{Institut f\"ur Theoretische Physik, Leibniz Universit\"at Hannover, Appelstraße 2, 30167 Hannover, Germany\\ Okinawa Institute of Science and Technology Graduate University, 1919-1 Tancha, Onna-son, Okinawa 904-0495, Japan}
\author{Giovanni Natale}
\email{giovanninatale96@gmail.com}
\affiliation{Okinawa Institute of Science and Technology Graduate University, 1919-1 Tancha, Onna-son, Okinawa 904-0495, Japan}
\author{Christopher S.P. Wever}
\email{christopher.wever@de.bosch.com}
\affiliation{Corporate Sector Research and Advance Engineering, Robert Bosch GmbH, Robert-Bosch-Campus 1, D-71272 Renningen, Germany}
\author{Philipp A. H\"ohn}
\email{philipp.hoehn@oist.jp}
\affiliation{Okinawa Institute of Science and Technology Graduate University, 1919-1 Tancha, Onna-son, Okinawa 904-0495, Japan}

\begin{abstract}
    Observables in quantum mechanics are generally complementary, that is, they reveal mutually incompatible pieces of information about a given system. This property is not only a fundamental tenet of the quantum formalism, but also a key component in quantum cryptographic protocols. As such it has fuelled much research into finding sets of highly complementary observables. In its strongest form -- the one we consider in this work -- the information between complementary observables is not merely incompatible but mutually exclusive: maximal information about one observable implies no information about the other, and vice versa. Maximal sets of non-degenerate complementary observables are known to exist in systems of prime power dimension such as $n$-qubit systems. Here, we study complementarity of degenerate observables, specifically we prove that observables associated with the $n$-qubit Pauli group also exhibit complementarity under this restriction: first, we obtain a criterion for two such observables to be complementary and, second, we relate maximal sets of complementary observables with informational pure state complementarity equalities, further studied in two companion papers. Equivalently, these results can be formulated in terms of maximal complementary sets of (not necessarily maximal) Abelian subgroups of the Pauli group linking with (possibly coarse-grained) mutually unbiased bases. Finally, we prove that these complementarity sets entail strong entropic uncertainty relations.
\end{abstract}

\maketitle

\section{Introduction}

Heisenberg's uncertainty relation \cite{Heisenberg1927,Kennard1927} remains a hallmark of quantum mechanics. As rightly anticipated by Bohr \cite{Bohr1928}, it foreshadows the dramatic departure from the edifice of classical physics - a development that has been unfolding over much of the last century, and keeps unfolding still. He coined the term `complementarity' to stand for the fundamental impossibility to jointly describe different aspects of a quantum system. Beyond its foundational value, this complementary nature of quantum measurement underlies many quantum cryptographic protocols, e.g.\ it limits the ability of an eavesdropper to interfere with the generation of a secret key between two parties \cite{BennettBrassard1984,BennettBrassardMermin1992}, enables to lock information that can be unlocked by sending a limited key \cite{DiVindenzoEtAl2004,ChristandlWinter2005,KoenigEtAl2007}, and underlies variants of secure bit commitment \cite{DamgardEtAl2007}. Unlike their classical counterparts, such protocols do not rely on computational hardness assumptions, but merely on the validity of quantum theory and the complementarity at its core.

A crucial ingredient to (prove the security of) these and related protocols are \emph{entropic uncertainty relations} \cite{WehnerWinter2010,BialynickiBirulaRudnicki2011,BertaEtAl2010,ColesEtAl2017,PortmannRenner2022}, which quantify the complementarity in a set of measurements. Strong bounds on entropic uncertainty relations hold between observables that are mutually unbiased \cite{MaassenUffink1988,SanchezRuiz1994,SanchezRuiz1995,Siudzinska2022}. Such observables are \emph{maximally complementary} in the sense that certainty about one implies maximal uncertainty about the other, and vice versa.

Here, we revisit this strong notion of complementarity, and how it is realised in quantum theory in the form of non-degenerate observables and mutually unbiased bases. We then consider complementarity under coarse-graining. Specifically, we look for instances of coarse-grained complementary observables associated with the $n$-qubit Pauli group: first, we consider the case of two such observables, and obtain a criterion for their complementarity (Thm.~\ref{thm: complementary Abelian subgroups}); second, we consider maximal sets of complementary degenerate observables. 

Equivalently, these results can be formulated using Abelian subgroups of the Pauli group to which these observables are naturally associated, leading to the notion of maximal sets of complementary Abelian subgroups. Our main result (Thm.~\ref{thm: max complementary sets}) relates such sets with certain quadratic purity invariants, which constitute informational pure state complementarity equalities (see Eq.~(\ref{eq: BZ info measure})).  Together with our parallel results in Ref.~\cite{FrembsHoehnNataleWever2026a,Frembs2026}, where we study such invariants in more detail, this proves the existence of such complementary sets. Moreover, we prove that these give rise to strong bounds in entropic uncertainty relations (Thm.~\ref{thm: collision entropy bound}).

\section{Complementarity revisited}\label{sec: complementarity refined}

We first recall the notion of complementarity, on an abstract level. In its weakest form, complementarity refers to the fundamental impossibility to measure observables jointly, and as such corresponds with the fact that observables in quantum mechanics are represented by generally noncommuting self-adjoint operators (more generally, incompatible positive operator-valued measures) \footnote{In this work, we will only consider sharp observables.}. Entropic uncertainty relations involve quantitative bounds, measuring the extent to which observables in quantum mechanics are complementary. In its strongest form - the one we consider in this work - complementarity then assumes the following definition (not tied to quantum theory and invoked in the reconstruction of Ref.~\cite{Hoehn2017,HoehnWever2017}). 

\begin{definition}\label{def: complementarity}
    Two observables are called \emph{complementary} if certainty about one observable implies maximal uncertainty about the other. A set of observables is called \emph{(mutually) complementary} if each of its members is complementary to every other.
\end{definition}

The notion of (un)certainty, equivalently, maximal (minimal) information in Def.~\ref{def: complementarity} can be quantified using various notions of entropic (resp.\ information) measures (see e.g.\ App.~\ref{app: purity invariants} and App.~\ref{app: entropic UR}). Yet, the precise choice of measure is irrelevant as long as different measures agree on states of minimal and maximal uncertainty (see below) \footnote{This is part of any axiomatisation of entropy, e.g.\ \cite{Renyi1961,ColesEtAl2012}.}.

For the study of complementarity, it is therefore not necessary to introduce an explicit entropy (resp.\ information) measure. (However, we will introduce one when translating our results into entropic uncertainty bounds in App.~\ref{app: entropic UR}.) Instead, we will simply characterise minimal and maximal uncertainty states w.r.t.\ some observable $O$, as those which correspond with deterministic, respectively unbiased probability distributions. More precisely, given a (sharp) quantum observable $O=\sum_{x\in\cX}x\Pi_x$ with outcome set $\cX$, define with Born's rule the probability distribution over $\cX$ given a state $\rho$ by
\begin{align}\label{eq: observable distribution}
    p^O_\rho(x)
    =\tr[\rho \Pi_x]\quad\quad\forall x\in\cX\; ,
\end{align}
where $\Pi_x$ denotes the projector onto the eigenspace corresponding to the outcome $x\in\cX$. Then $\rho$ is called \emph{a state of maximal uncertainty (or unbiased) in $O$} if $p^O_\rho(x)=\frac{1}{|\cX|}$ is the uniform distribution. In contrast, $\rho$ is called \emph{a state of minimal uncertainty in $O$} if $p^O_\rho=\delta_{xx_0}$ for some $x_0\in\cX$ (and where $\delta_{xx_0}=1$ if $x=x_0$ and $\delta_{xx_0}=0$ otherwise).

It is not too hard to see that non-degenerate complementary observables, i.e.\ those with $|\cX|=d:=\dim(\cH)$, correspond with mutually unbiased bases (see Sec.~\ref{sec: complementary Pauli subgroups}), and this correspondence even allows one to construct maximal sets of $d+1$ complementary observables whenever $d$ is a prime power \cite{Ivanovic1981,WoottersFields1989}.

What about complementarity for degenerate observables, i.e.\ when $|\cX|<d$? One way to think about degenerate observables is in terms of coarse-grainings of non-degenerate ones. More precisely, given an observable $O=\sum_{x\in\cX}x\Pi_x$ and a coarse-graining function $f:\cX\ra\tX$, define the coarse-grained observables (with respect to $f$) by $\tO=f(O)=\sum_{x\in\cX}f(x)\Pi_x=\sum_{\tx\in\tX}\tx\Pi_\tx$, where $\Pi_\tx=\sum_{x\in f^{-1}(\tx)}\Pi_x$. Now, note that a state that is not a state of maximal or minimal uncertainty in (a non-degenerate observable) $O$, may still be a maximal or minimal uncertainty state in a coarse-grained observable $f(O)$. This indicates that complementarity weakens under coarse-graining. Moreover, when applied to degenerate observables, Def.~\ref{def: complementarity} poses more constraints than in the non-degenerate case, as there are more minimal uncertainty states in each observable which are required to be unbiased with respect to complementary ones (although there are fewer outcomes to which this applies). A priori, it is thus not clear that quantum mechanics exhibits complementarity in coarse-grained form. Here, we will show this to be the case. Specifically, we will analyse complementarity for degenerate observables associated with the $n$-qubit Pauli group. Indeed, generalised Pauli operators possess only two outcomes with equal degeneracy.

\section{Complementarity in the $n$-qubit Pauli group}\label{sec: complementary Pauli subgroups}

In general, complementarity between observables is a property of the relative geometry of their eigenspaces. A natural candidate geometry is that of mutually unbiased bases (MUB). Recall that two bases $\cB,\cB'$ of a Hilbert space $\cH$ of dimension $\dim(\cH)=d$ are called \emph{mutually unbiased} if $|\langle b|b'\rangle|^2=\frac{1}{d}$ for all $b\in\cB$ and $b'\in\cB'$. It follows immediately from Def.~\ref{def: complementarity} that collections of non-degenerate observables whose eigenstates define MUBs are complementary.  MUBs naturally arise in the form of stabiliser states of maximal sets of mutually commuting $n$-qubit Pauli operators (see App.~\ref{app: Pauli group}).

In fact, it turns out that for projective measurements diagonal with respect to a set of commuting $n$-qubit Pauli observables, complementarity reduces to the incidence structure of Abelian subgroups in the Pauli group. The focus on sets of Abelian sub\emph{groups} to define a generalised notion of complementarity for sets of compatible observables is operationally justified; knowing the outcomes to some set of commuting observables implies knowledge about the outcomes of the entire Abelian group generated by them. We will therefore define a notion of complementarity for sets of Abelian subgroups, such that maximal information about one comes at the expense of total ignorance about another.\\

\textbf{Complementary Abelian subgroups.} Let $\cP_n:=\langle\one,X,Y,Z\rangle^{\otimes n}$ be the $n$-qubit Pauli group. For an Abelian subgroup $A<\cP_n$, we say that an observable $O_A$ is \emph{associated with $A$} if it is diagonal in the joint eigenbasis of the elements in $A$ and has outcome space $|\cX|=2^m$. Note that in order to characterise minimal and maximal uncertainty states, the specific choice of outcome sets $\cX$ is irrelevant; consequently, up to this choice, coarse-graining for observables $O_\tA=f_{\tA,A}(O_A)$ is fixed by subgroup inclusion $\tA\subset A$. Moreover, since Born rule probabilities only depend on the projectors onto the eigenspaces of an observable, to study complementarity (as in Def.~\ref{def: complementarity}) it is sufficient to consider Pauli observables up to phase. Throughout, we will therefore restrict to the set of Hermitian Paulis $\aP_n:=\{\one,
X,Y,Z\}^{\otimes n}$ (see App.~\ref{app: Pauli group}), and write $\cA_{(m)}(\aP_n)$ for the set of Abelian (stabiliser) subgroups, that is, Abelian subgroups $A<\cP_n$  with $\one\notin A\subset \pm\aP_n$ (of dimension $0\leq m\leq n$, that is, isomorphic to $\zz^m_2$). Disregarding phases, we will also write $A<\aP_n$, for short.

\begin{definition}\label{def: complementary subgroups}
    A set of Abelian subgroups $\{A_k\}_k\subset\cA(\aP_n)$ is \emph{complementary} if and only if any set of associated observables $\{O_{A_k}\}_k$ is (mutually) complementary \footnote{We use `complementary set of Abelian subgroups' and `set of complementary Abelian subgroups' synonymously.}.
\end{definition}

We begin to characterise complementarity for two (not necessarily maximal) Abelian subgroups in $\aP_n$.\\

\textbf{Two complementary subgroups.} For a maximal Abelian subgroup $A\in\cA_n(\aP_n)$, let $\{|\psi^{A,x}\rangle\}_{x\in\zz_2^n}$ be the joint eigenbasis (of stabiliser states) of $A$ such that the elements $a\in A$ can be labelled by bit-strings $y\in\zz^n_2$ via
\begin{align}\label{eq: Pauli representation}
    a(y)
    &=\sum_{x\in\zz_2^n} (-1)^{x\cdot y} \dyad{\psi^{A,x}}\; ,
\end{align}
where $x\cdot y:=\sum_{i=1}^nx_iy_i$. Note that $a(0)=\one$. In other words, $|\psi^{A,x}\rangle$ is the unique stabiliser state with stabiliser group $A(x):=\{(-1)^{x\cdot y}a(y)\mid y\in\zz_2^n\}\in\cA_n(\aP_n)$, i.e.\ $A(x)|\psi^{A,x}\rangle=|\psi^{A,x}\rangle$ for all $x\in\zz^n_2$. Recall that two maximal Abelian subgroups are complementary if and only if their respective bases of eigenstates are mutually unbiased. The following lemma translates this into a statement on the level of maximal Abelian subgroups. 

\begin{lemma}[\cite{AaronsonGottesman2004,KuengGross2015}]\label{lm: AG}
    Let $\{|\psi^{A,x}\rangle\}_{x\in\zz^n_2}$, $\{|\psi^{A',x'}\rangle\}_{x'\in\zz^n_2}$ be the joint eigenbases of two maximal Abelian subgroups $A,A'\in\cA_n(\aP_n)$ such that Eq.~(\ref{eq: Pauli representation}) applies to all $a(y)\in A\cap A'$. Then $|\langle\psi^{A,x}|\psi^{A',x'}\rangle|^2=\frac{|A\cap A'|}{2^n}$ if $x\cdot y=x'\cdot y$ for all $a(y)\in A\cap A'$, and $|\langle\psi^{A,x}|\psi^{A',x'}\rangle|^2=0$ otherwise.
\end{lemma}

From this we obtain the following condition for two maximal Abelian subgroups to be complementary.

\begin{corollary}\label{cor: maximal complementary Abelian subgroups}
    Two maximal Abelian subgroups $A,A'\in\cA_n(\aP_n)$ are complementary if and only if $A\cap A'=\{\one\}$.
\end{corollary}

\begin{proof}
    A state of maximal information in $A$ is necessarily an eigenstate of $A$, which is uniform in $A'$ if and only if $|A\cap A'|=\{\one\}$ by Lm.~\ref{lm: AG} (similarly for $A\leftrightarrow A'$).
\end{proof}

However, this condition is not sufficient to characterise complementarity of non-maximal Abelian subgroups.

\begin{proposition}\label{prop: complementary measurments}
    Two Abelian subgroups $\tA,\tB<\aP_n$ with $\tA\cap\tB=\{\one\}$ are not complementary, in general.
\end{proposition}

\begin{proof}
    Clearly, every state $|\psi^A\rangle$ with stabiliser $\tA\subset A\in\cA_n(\aP_n)$ is a state of minimal uncertainty in $\tA$ (cf. Lm.~\ref{lm: max info states}). Yet, if $A\cap\tB\neq\{\one\}$, then $|\psi^A\rangle$ is not unbiased in $\tB$, in which case $\tA$ and $\tB$ are not complementary.
\end{proof}

Let $C(\tA)=\{P\in\aP_n\mid [P,Q]=0\ \forall Q\in\tA\}$ be the commutant of $\tA$ in $\aP_n$, and note that as a consequence of Prop.~\ref{prop: complementary measurments} it makes sense to require the following necessary condition: two Abelian subgroups $\tA,\tB<\aP_n$ are complementary only if $\tA\cap C(\tB)=\{\one\}=\tB\cap C(\tA)$. Clearly, this holds for maximal Abelian subgroups $A,B\in\cA_n(\aP_n)$, since in that case $C(A)=A$ and $C(B)=B$. Our first result proves that this condition is also sufficient, hence, characterises complementarity for any (not necessarily maximal) Abelian subgroups in the $n$-qubit Pauli group.

\begin{theorem}\label{thm: complementary Abelian subgroups}
     Let $\tA,\tB<\aP_n$ be two (not necessarily maximal) Abelian subgroups. Then $\tA$ and $\tB$ are complementary if and only if $\tA \cap C(\tB) =\{\one\}$ and $\tB \cap C(\tA) =\{\one\}$.
\end{theorem}

\begin{proof}
    We provide the details in App.~\ref{app: proof-two observables}
\end{proof}

Thm.~\ref{thm: complementary Abelian subgroups} applies to Abelian subgroups that are not necessarily of the same order. Since $C(\tA)<C(\tA')$ for $\tA'<\tA$, it implies that simply coarse-graining an existing complementary set is not guaranteed to yield another complementary set. In particular, complementarity for degenerate observables is a much stronger condition than the mere coarse-grained notion of mutual unbiasedness $\tr[\Pi^\tA_x\Pi^\tB_y]=\mathrm{const}$ for projectors $\Pi^\tA_x$ and $\Pi^\tB_y$ onto subsets of eigenstates corresponding to Abelian subgroups $\tA,\tB\in\cA_m$ that arise from coarse-graining of a set of mutually unbiased bases.\\

\textbf{Maximal sets of complementary subgroups.} Having analysed complementarity for two Abelian subgroups in the $n$-qubit Pauli group $\aP_n$ in Thm.~\ref{thm: complementary Abelian subgroups}, we turn to the other extreme, that is, we consider maximal (in the sense of largest unextendible, see below) sets of complementary Abelian subgroups in $\aP_n$. By Lm.~\ref{lm: AG} and Cor.~\ref{cor: maximal complementary Abelian subgroups}, for non-degenerate observables this corresponds with the existence of maximal sets of mutually unbiased bases (MUBs). While the problem of finding maximal sets of MUBs in general Hilbert space dimension $d=\dim(\cH)$ remains open, it is well-known that in dimension $\dim(\cH)=d=p^n$ for $p$ prime and $n\in\mathbb
N$, a maximum of $d+1$ MUBs exist \cite{WoottersFields1989}. Moreover, maximal sets of MUBs can be constructed from stabiliser states. In terms of the corresponding maximal Abelian subgroups, these correspond with partitions $\cM=\{A_k\}_{k=1}^{2^n+1}$ of $\aP^\circ_n=\bigcupdot_{k=1}^{2^n+1}A^\circ_k$, where $\cS^\circ:=\cS\backslash\{\one\}$ for any $\one\in\cS\subseteq\aP^n$, of the $n$-qubit Pauli group into $2^n+1$ maximal Abelian subgroups.

\begin{theorem}[\cite{BandyopadhyayVatan2002,LawrenceBruknerZeilinger2002}]\label{thm: LBZ}
    Every partition $\cM=\{A_k\}_{k=1}^{2^n+1}$ of $\aP^\circ_n=\bigcupdot_{k=1}^{2^n+1}A^\circ_k$ into maximal Abelian subgroups $A_k\in\cA_n(\aP_n)$ defines a maximal set of $2^n+1$ mutually unbiased bases consisting of stabiliser states $\{|\psi^{A,v}\rangle\}_{v\in\zz^n_2,A\in\cM}$.
\end{theorem}

We remark that $\aP_n$ admits many different, including unitarily inequivalent partitions \cite{SehrawatKlimov2014}.

\begin{corollary}\label{cor: max complementary max subgroups}
    Maximal complementary sets of maximal Abelian subgroups in $\aP_n$ correspond with partitions of $\aP_n$.
\end{corollary}

\begin{proof}
    Clearly, every maximal complementary set of maximal Abelian subgroups defines a partition of $\aP_n$. Conversely, it follows from Thm.~\ref{thm: LBZ}, and the properties of MUBs, that every partition defines a maximal complementary set of maximal Abelian subgroups in $\aP_n$.
\end{proof}

Maximal sets of complementary subgroups give rise to strong bounds for entropic uncertainty relations \cite{BallesterWehner2007,WehnerWinter2010}. This motivates to study sets of complementary Abelian subgroups that are not necessarily of maximal dimension $n$. In general, we may allow these subgroups to be of different orders, in which case, states of minimal uncertainty corresponding to different complementary observables will have different total entropy.

Here, we restrict to the case where each observable contains the same amount of information, that is, we look for maximal sets of Abelian subgroups $\{\tA_k\}_{k=1}^K\subset\cA_m(\aP_n)$ for fixed $1\leq m\leq n$. By Cor.~\ref{cor: max complementary max subgroups} and Thm.~\ref{thm: LBZ}, maximal complementary sets of Abelian subgroups of dimension $n$ in $\aP_n$ have cardinality $K=2^n+1$. In fact, this is the maximum number for any complementary set of Abelian subgroups of fixed dimension.

\begin{lemma}\label{lm: maximality bound}
    Let $\{\tA_k\}_{k=1}^K$ be a complementary set of Abelian subgroups of dimension $1\leq m\leq n$, then $K\leq 2^n+1$.
\end{lemma}

\begin{proof}
    We provide the proof in App.~\ref{app: proof-max complementary sets}.
\end{proof}

Clearly, not every set (of complementary subgroups) is maximal. If $\tM$ is such a (non-empty) set, then $\tM'\subset\tM$ is still a complementary set after removing any number of Abelian subgroups in $\tM$. Yet, not every complementary set arises in this way from a maximal one, as it is not generally possible to extend a non-maximal complementary set to a larger one. For instance, complementary sets $\tM=\{\tA_k\}_{k=1}^K$ with $\tA_k=\langle Q_k\rangle$ and $Q_k\in\aP_n$ for every $k\in[K]$ correspond with maximal sets of mutually anti-commuting Pauli operators, which have cardinality at most $2n+1$ (see Lm.~8 and Cor.~4 in Ref.~\cite{SarkarVanDenBerg2021}). It is generally non-trivial to decide whether a complementary set can be extended, as is clear from the construction of so-called unextendible mutually unbiased bases in Ref.~\cite{MandayamEtAl2014,Thas2016}. Accordingly, we also call a complementary set $\tM=\{A_k\}_{k=1}^K\subset\cA_m(\aP_n)$ of Abelian subgroups \emph{unextendible} if there exists no complementary set $\tM'=\{A'_k\}_{k=1}^{K'}\subset\cA_m(\aP_n)$ with $K'>K$ and $\tM\subset\tM'$. Evidently, unextendibility is a weaker notion of `maximality' for complementary sets: complementary sets with $K=2^n+1$ are maximal among all unextendible ones.

Given the highly restrictive conditions in Thm.~\ref{thm: complementary Abelian subgroups}, one may expect maximal sets of complementary subgroups to exist only in the maximal Abelian case, corresponding with non-degenerate observables. Nevertheless, our main result implies that this is not the only case. To state it, define the \emph{Brukner-Zeilinger information measure} \cite{BruknerZeilinger1999,BruknerZeilinger2001,Hoehn2017,HoehnWever2017} (see also App.~\ref{app: purity invariants}) of a subset $\PI\subset\aP_n$,
\begin{align}\label{eq: BZ info measure}
    I_\rho(\PI)
    :=\sum_{P\in\PI}\tr[P\rho]^2\; ,
\end{align}
for any density operator $\rho\in\cD(\C^{2^n})$, and call a subset $\PI$ a \emph{purity invariant} if $I_\psi(\PI)=\mathrm{const}$ for all $\psi\in\CP^{2^n-1}$.

\begin{theorem}\label{thm: max complementary sets}
    Every maximal complementary set of Abelian subgroups $\{\tA_k\}_{k=1}^{2^n+1}$ defines a purity invariant, that is, $I_\psi(\bigcupdot_{k=1}^{2^n+1}\tA^\circ_k)=2^m-1$ for every $\psi\in\CP^{2^n-1}$.
\end{theorem}

\begin{proof}
    We provide the details in App.~\ref{app: proof-max complementary sets}.
\end{proof}

We may thus view the purity invariants of Thm.~\ref{thm: max complementary sets} as informational pure-state complementarity equalities; they encode that maximal information about (the observables in) one Abelian subgroup implies total ignorance about the others in a maximal complementarity set.\footnote{In turn, the relation between qubit complementarity and purity invariants in Thm.~\ref{thm: max complementary sets} constitutes a key motivation to our parallel work in Ref.~\cite{FrembsHoehnNataleWever2026a}, where we generalise the two-qubit purity invariants, first discovered in Ref.~\cite{HoehnWever2017}.}

Clearly, $\aP_n$ defines a purity invariant. More interestingly, purity invariants exist for $m\neq n$, in particular, the family of purity invariants defined in Ref.~\cite{FrembsHoehnNataleWever2026a} constitutes maximal complementary sets for $m=n-1$. Explicit examples for $n=2,3$, which can be constructed using MUB partitions according to Thm.~\ref{thm: LBZ}, are provided in \cite{FrembsHoehnNataleWever2026a}.

\begin{corollary}\label{cor: existence of non-trivial maximal complementary sets}
    Maximal complementary sets of Abelian subgroups $\{\tA_k\}_{k=1}^{2^n+1}$ exist for $m=n-1$ and $n\geq 1$.
\end{corollary}

\begin{proof}
    We provide the details in App.~\ref{app: proof-max complementary sets}.
\end{proof}

Thm.~\ref{thm: max complementary sets} and Cor.~\ref{cor: existence of non-trivial maximal complementary sets} show that beyond the case of maximal MUB sets, the $n$-qubit Pauli group hides other complementarity structures. The existence of these complementarity structures and their closely related invariants is surprising and - in light of Thm.~\ref{thm: complementary Abelian subgroups} - hints at a rich algebraic structure of these sets. In particular, the complementary sets corresponding to the purity invariants in Ref.~\cite{FrembsHoehnNataleWever2026a} are closed under anti-commutators (cf.~Ref.~\cite[Lm.~3]{FrembsHoehnNataleWever2026a}). Under this constraint, the proof of Cor.~\ref{cor: maximal complementary Abelian subgroups} in App.~\ref{app: proof-max complementary sets} further yields a converse to Thm.~\ref{thm: max complementary sets}: every anti-commutator closed purity invariant gives rise to a maximal complementary set of Abelian subgroups. Moreover, in Ref.~\cite{Frembs2026} it is shown that anti-commutator closed purity invariants are precisely those of the family in Ref.~\cite{FrembsHoehnNataleWever2026a}, as well as that maximal complementary sets and purity invariants exist that do not satisfy this property.

We close with another important consequence of the close relationship between complementarity and informational invariants in Thm.~\ref{thm: max complementary sets}: complementary subgroups have entropic uncertainty relations with strong bounds.

\begin{theorem}\label{thm: collision entropy bound}
    Let $\{\tA_k\}_{k=1}^{2^n+1}$ be a maximal complementary set of Abelian subgroups. Then the collision entropy (see App.~\ref{app: entropic UR}) satisfies the following uncertainty relation,
    \begin{align}\label{eq: collision entropy bound}
        \frac{1}{2^n+1}\sum_{k=1}^{2^n+1}H_2(\tA_k\mid\psi)
        \ \geq\
        \log(\frac{2^n+1}{2^{n-m}+1})\; .
    \end{align}
\end{theorem}

\begin{proof}
    We prove this in App.~\ref{app: entropic UR}.
\end{proof}

\section{Outlook}\label{sec: outlook}

We studied complementarity (as by Def.~\ref{def: complementarity}) in the $n$-qubit Pauli group $\cP_n$. While the case of non-degenerate observables is well known to correspond with partitions of $\cP_n$ into maximal Abelian subgroups (equivalently, mutually unbiased bases), we showed that $\cP_n$ also gives rise to maximal complementary sets of degenerate observables, which are associated with non-maximal Abelian subgroups. These sets reveal previously unnoticed algebraic structure in the $n$-qubit Pauli group, which we analyse in more detail in two companion papers, Ref.~\cite{FrembsHoehnNataleWever2026a,Frembs2026}.

While our study showcases that quantum theory allows for coarse-grained notions of complementarity (as by Def.~\ref{def: complementarity}), it can only mark the beginning of a more systematic study. Some immediate open questions include: Can one characterise all maximal complementary sets, more generally, all (non-maximal) unextendible complementary sets of Abelian subgroups in the $n$-Pauli group? What about the $n$-qudit Pauli group in (prime) dimension $d>2$ \cite{PlanatSaniga2008}? To answer these and similar questions, it would be interesting to translate our results into the language of incidence geometry of symplectic vector spaces \cite{PlanatSaniga2008,Thas2009,MandayamEtAl2014,Thas2016}. Going beyond the Pauli group, are there other maximal complementary sets of observables, e.g. mutually unbiased measurements (`MUMs') \cite{KalevGour2014} or (symmetrically) information-complete POVMs \cite{Siudzinska2022} that also exhibit coarse-grained forms of complementarity with respect to Def.~\ref{def: complementarity}? Even if not, what about weaker forms of complementarity, e.g. approximate versions, and the related bounds on entropic uncertainty relations \cite{WehnerWinter2010}?

Beyond its foundational value, complementarity also constitutes an essential ingredient in proofs of security of modern quantum cryptographic applications \cite{BennettBrassard1984,DiVindenzoEtAl2004,BertaEtAl2010,MancinskaStorgaard2022}; namely, in the form of entropic uncertainty relations, which for instance bound the amount of information an eavesdropper can obtain in quantum secret key generation \cite{BennettBrassard1984}, or limits the extent to which information can be locked \cite{DiVindenzoEtAl2004}. In both cases, one seeks complementary measurements with strong bounds on their entropic uncertainty \cite{WehnerWinter2010,ColesEtAl2017}, suggesting that coarse-grained complementarity in Thm.~\ref{thm: max complementary sets} (and the uncertainty relations in Thm.~\ref{thm: collision entropy bound}) may find use in such cryptographic protocols.

\textbf{Acknowledgments.} This work was made possible through the support of the ID\# 62312 grant from the John Templeton Foundation, as part of the project \href{https://www.qiss.fr/}{``The Quantum Information Structure of Spacetime'' (QISS)} and through the \href{https://withoutspacetime.org}{``WithOut SpaceTime project'' (WOST)}, led by the Center for Spacetime and the Quantum (CSTQ), and funded by Grant ID\# 63683 from the John Templeton Foundation. The opinions expressed in this work are those of the authors and do not necessarily reflect the views of the John Templeton Foundation.

\bibliography{bibliography}

\appendix
\onecolumngrid

\section{Pauli group}\label{app: Pauli group}
We will consider complementary measurements built from $n$-qubit Pauli observables. This section contains mostly standard facts that will be used later, and serves to introduce our notation (see also Ref.~\cite{FrembsHoehnNataleWever2026a,Frembs2026}).

The $n$-qubit Pauli group $\cP_n=\langle\one,X,Y,Z\rangle^{\otimes n}$ is generated by the standard single-qubit Pauli matrices:
\begin{align*}
    \one&=\begin{pmatrix}
        1 & 0 \\ 0 & 1
    \end{pmatrix} &
    X&=\begin{pmatrix}
        0 & 1 \\ 1 & 0
    \end{pmatrix} &
    Y&=\begin{pmatrix}
        0 & -i \\ i & 0
    \end{pmatrix} &
    Z&=\begin{pmatrix}
        1 & 0 \\ 0 & -1
    \end{pmatrix}
\end{align*}
$\cP_n$ is a group with center $K=Z(\cP_n)=\{\pm 1,\pm i\}$. The space of Hermitian Pauli observables is conveniently parametrised in terms of Weyl operators: let $v=(a,b)\in(\zz_2^n\times\zz_2^n)=\zz^{2n}_2=V$ and define
\begin{align}\label{eq: Weyl operators}
    W_v
    =i^{a \cdot b}(Z^{a_1} \otimes\cdots\otimes Z^{a_n})(X^{b_1} \otimes\cdots\otimes X^{b_n})\; ,
\end{align}
where $a\cdot b:=\sum_{i=1}^n a_ib_i$. Note that $\aP_n=\{\one,X,Y,Z\}^{\otimes n}\cong (W_v)_{v\in\zz^{2n}_2}$, whereas $\cP_n=(\zeta W_v)_{v\in\zz^{2n}_2}$ for $\zeta\in K=\{\pm 1,\pm i\}$. Crucially, the space of Weyl operators $(W_v)_{v\in\zz_2^n}$ comes equipped with a \emph{symplectic structure $\omega:V\times V\ra\zz_2$} given by $\omega(v_1,v_2)=a_1b_2-b_1a_2$.\footnote{In matrix notation, $\omega(v,w)=v^T\sigma_n w$, $\sigma_n=\begin{pmatrix}
    0_n & \one_n \\
    -\one_n & 0_n
\end{pmatrix}$, for $0_n$ ($\one_n$) the additive (multiplicative) unit in $M_n(\C)$.} The commutation relations are then neatly expressed as follows:
\begin{align}\label{eq: Weyl commutation relations}
    W_vW_w=(-1)^{\omega(v,w)} W_wW_v \quad \forall v,w \in V\; ,
\end{align}
in particular, $[W_v,W_w]=0$ if and only if $\omega(v,w)=0$.\footnote{Moreover, in the qubit case, this is equivalent to $\{W_v,W_w\}=0$ if and only if $\omega(v,w)=1$.} Yet, complementarity in the Pauli group only depends on the (anti-)commutation relations, hence, is independent of any phase conventions (e.g. those in the Weyl representation in Eq.~(\ref{eq: Weyl operators}) below). We will thus mostly restrict to the Hermitian Pauli observables $\aP_n=\{\one,X,Y,Z\}^{\otimes n}$, and encode their commutation relations via the map $\omega:\aP_n\times\aP_n\ra\zz_2$ defined by $\omega(Q,P)=0$ if and only if $[Q,P]=0$.

A \emph{symplectic transformation $S:V\ra V$} is a linear map that preserves the symplectic product, i.e., $\omega(Sv,Sw)=\omega(v,w)$ for all $v,w\in V=\zz^{2n}_2$. We denote the group of symplectic transformations by $\mathrm{Sp}_{2n}(\zz_2)$. Moreover, the normaliser of the $n$-qubit Pauli group is the Clifford group $\Cl_n:=\{U\in\cU(2^n)\mid U\cP_nU^\dagger\in\cP_n\}$.

We record the following standard result about the relation between these two groups \cite{NebeRainsSloane2001,DehaeneDeMoor2003,deBeaudrap2013}.

\begin{theorem}\label{thm: Pauli automorphism group}
    $\mathrm{Sp}_{2n}(\zz_2)=\Cl_n/(U(1)\cP_n)$.
\end{theorem}

\textbf{Commutants.} We write $\cA_{(m)}(\aP_n)$ for the set of Abelian (stabiliser) subgroups, that is, Abelian subgroups $A<\cP_n$ with $\one\notin A\subset \pm\aP_n$ (of dimension $0\leq m\leq n$, that is, isomorphic to $\zz^m_2$), as well as $A<\aP_n$ for $A\in\cA(\aP_n)$. A linear subspace $V\subset\zz^{2n}_2$ is called \emph{isotropic} if $\omega(v,v')=0$ for all $v,v'\in V$.\footnote{Equivalently, $V \subset V^\perp:=\{w\in\zz_2^n \mid \forall v\in V:\ \omega(w,v)=0\}$. A maximal isotropic subspace is also called \emph{Lagrangian}.} Under the correspondence in Eq.~(\ref{eq: Weyl commutation relations}), it follows immediately that Abelian subgroups $A<\aP_n$ correspond with isotropic subspaces. We say that $\cG\subset\aP_n$ is a generating set of an Abelian subgroup $A<\aP_n$ if $A=\langle\cG\rangle$. For every Abelian subgroup $A<\aP_n$, let $C(A)=\{Q\in\aP_n\mid [P,Q]=0\ \forall Q\in A\}$ denote the commutant of $A$ in $\aP_n$, and let $C_\cS(A)=C(A)\cap\cS$ for any subset $\cS\subset\aP_n$.

\begin{lemma}\label{lm: commutants are product closed}
    Let $\cS\subset\aP_n$ be a set of mutually commuting $n$-qubit Pauli operators. Then $C(\cS)=C(\langle\cS\rangle)$. Let $A<\aP_n$ be an Abelian subgroup. Then its commutant $C(A)$ is closed under products.
\end{lemma}

\begin{proof}
    Both statement follow immediately from bilinearity of the symplectic form $\omega$.
\end{proof}

We will further use the following lemma, adapted from Lm.~3 in Ref.~\cite{SarkarVanDenBerg2021}.

\begin{lemma}[\cite{SarkarVanDenBerg2021}]\label{lm: commutativty maps}
    Let $A\in\cA_m(\aP_n)$ be an Abelian subgroup of dimension $m$. Then $|C(A)|=\frac{4^n}{2^m}$.
\end{lemma}

In particular, every element $Q\in\aP_n$ in the $n$-qubit Pauli group (anti-)commutes with $\frac{1}{2}4^n$ elements in $\aP_n$.

The next lemma generalises Thm.~\ref{thm: LBZ} to commutants.

\begin{lemma}\label{lm: commutant MUB sub-partition}
    Let $A\in\cA_m(\aP_n)$ be an Abelian subgroup of dimension $m$. Then $C(A)/A$ admits a partition into $2^{n-m}+1$ Abelian subgroups of cardinality $2^{n-m}-1$ (each with the identity element removed).
\end{lemma}

\begin{proof}
    Without loss of generality, let $\cG=\{P_1,\cdots,P_m\}$ be a generating set for $A=\langle\cG\rangle$ with $P_i=\one^{i-1}\otimes X_i\otimes\one^{n-i}$ for $1\leq i\leq m$. Clearly, in this case $C(A)=A\otimes\aP_{n-m}$, and by Thm.~\ref{thm: LBZ} there exists a partition of $\aP^\circ_{m-n}$ into $2^{n-m}+1$ Abelian subgroups $\{B_i\}_{i=0}^{2^{n-m}}$ such that $B_i\cap B_j=A$ for all $0\leq i\neq j\leq 2^{n-m}$. The general case follows since any Abelian subgroup $A<\aP_n$ is Clifford conjugate to one with a generating set as above (see also Lm.~\ref{lm: conjugate pairs}).
\end{proof}

\textbf{Normal form.} The next lemma shows that we can always bring two complementary maximal Abelian subgroups in normal form, corresponding with a change to a standard symplectic basis in the symplectic vector space $(V=\zz^{2n}_2,\omega)$.

\begin{lemma}\label{lm: conjugate pairs}
    Let $A,B\in\cA_n(\aP_n)$ be two maximal Abelian, complementary subgroups, that is, $A\cap B=\{\one\}$. Then there exists a Clifford unitary $U\in\Cl_n$ such that
    \begin{align*}
        UAU^\dagger&=A_Z=\langle Z_i\rangle_{i=1}^n &
        UBU^\dagger&=A_X=\langle X_i\rangle_{i=1}^n\; .
    \end{align*}
\end{lemma}

\begin{proof}
    Complementarity, $A\cap B=\{\one\}$, implies that we can order generating sets $\cG=\{W(v_i)\}_{i=1}^n$ for $A=\langle\cG\rangle$ and $\cH=\{W(w_i)\}_{i=1}^n$ for $B=\langle\cH\rangle$ such that $\omega(v_i,w_j)=\delta_{ij}$, that is, $(v_1,\cdots,v_n,w_1,\cdots,w_n)$ is a standard basis of the symplectic vector space $\zz^{2n}_2$. There exists a symplectic transformation $S\in\mathrm{Sp}_{2n}(\zz_2)$ such that $Sv_i=z_i$ and $Sw_i=x_i$ where $z_i=(\underbrace{0,\cdots,0}_{i-1},1,\underbrace{0,\cdots,0}_{2n-i-1})$ and $x_i=(\underbrace{0,\cdots,0}_{n+i-1},1,\underbrace{0,\cdots,0}_{n-i-1})$. The result follows since $\mathrm{Sp}_{2n}(\zz_2)<\Cl_n$ by Thm.~\ref{thm: Pauli automorphism group}.
\end{proof}

More generally, Lm.~\ref{lm: conjugate pairs} holds for subgroups satisfying the canonical commutation relations in Eq.~(\ref{eq: Weyl commutation relations}), as a consequence of the discrete version of (generalisations of) the Stone-von Neumann theorem \cite{Mackey1949}.

\begin{lemma}\label{lm: pairwise extendability}
    Any two maximal Abelian, complementary subgroups $A,B\in\cA_n(\aP_n)$ can be extended to a partition of $\aP^\circ_n=\bigcupdot_{i=0}^{2^n}A^\circ_i$ in maximal Abelian subgroups $A^\circ_i\in\cA_n(\aP_n)$.
\end{lemma}

\begin{proof}
    This follows from Lm.~\ref{lm: conjugate pairs}, the relation between maximal sets of $2^n+1$ mutually unbiased bases and partitions of $\aP^\circ_n$ in Ref.~\cite{LawrenceBruknerZeilinger2002}, and the explicit construction of such mutually unbiased bases in Ref.~\cite{Ivanovic1981,WoottersFields1989}.
\end{proof}

\begin{lemma}\label{lm: extension of two complementry subgroups}
    Let $\tA,\tB<\aP_n$ be two Abelian subgroups such that $\tA\cap\tB=\{\one\}$. Then there exist maximal Abelian subgroups $A,B\in\cA_n(\aP_n)$ such that $\tA<A$, $\tB<B$ and $A\cap B=\{\one\}$.
\end{lemma}

\begin{proof}
    Similar to Lm.~\ref{lm: conjugate pairs}, let $\cG=\{W(v_i)\}_{i=1}^k$ and $\cH=\{W(w_i)\}_{i=1}^l$ be generating sets for $\tA=\langle\cG\rangle$ and $\tB=\langle\cH\rangle$. After basic row operations, we can first bring the matrix $\omega(v_i,w_j)$ into canonical form $\omega(v_i,w_j)=\delta_{ij}$, and then extend it to a full symplectic basis of $\zz^{2n}_n$ (in canonical form), using a variant of Gram-Schmidt orthogonalisation.
\end{proof}

\textbf{Frustration graph.} The (anti-)commutation relations in $\aP^\circ_n$ can also be encoded in the following graph.

\begin{definition}\label{def: frustration graph}
    The \emph{frustration graph $G(n)=G(\aP^\circ_n)=(V(\aP^\circ_n),E(\aP^\circ_n))$ of $\aP^\circ_n$} is defined by $V(\aP^\circ_n):=\aP^\circ_n$ and $E(\aP^\circ_n):=\{(P,Q)\in\aP^\circ_n\times\aP^\circ_n\mid[P,Q]=0\}$.
\end{definition}

Recall that the \emph{degree of a vertex $v$ in $G=(V,E)$} is the number of its neighbours, $|\{v'\in V\mid (v,v')\in E\}|$.

\begin{definition}\label{def: bi-regular cuts}
    A \emph{cut $G(V,\overline{V})$ of $G(n)=G(\aP^\circ_n)$} is the bipartite graph with $V\subset V(\aP^\circ_n)$ arbitrary, $V\cupdot\overline{V}=V(\aP^\circ_n)$, and edges given by $E=\{(v,\overline{v})\in E(\aP^\circ_n)\mid v\in V,\overline{v}\in\overline{V}\}$.

    A cut $G(V,\overline{V})$ of $G(n)$ is called \emph{bi-regular} if all vertices $v\in V$ and $\overline{v}\in\overline{V}$ have the same degree.
\end{definition}

For the proof of our main result, Thm.~\ref{thm: max complementary sets}, we prove that bi-regular cuts of the frustration graph of the $n$-qubit Pauli group give rise to purity invariants as defined in Eq.~(\ref{eq: BZ info measure}).

\begin{lemma}[\cite{FrembsHoehnNataleWever2026a}]\label{lm: Pauli identity}
    For any pure state $|\psi\rangle\in\mathbb{CP}^{2^n-1}$ with corresponding Bloch vector representation $\rho = \dyad{\psi} = \frac{1}{2^n}(\one + \sum_{\one \neq P \in \aP_n} \alpha_P P)$, and for every Pauli operator $Q \in \aP_n$ the following identity holds,
    \begin{equation}\label{eq: Pauli identity}
        (2^{n-1}-1)(\alpha^2_Q+1) = \sum_{\one,Q \neq P\in C(Q)} \alpha^2_P\; .
    \end{equation}
\end{lemma}

\begin{proof}
    Let $\tau_Q:=\frac{1}{2^n}\left(\one+\sum_{\one\neq P\in C(Q)} \alpha_P P\right)$. Then $\tau_Q$ commutes with $Q$. Moreover, $\rho-\tau_Q$ anti-commutes with $Q$ (since two Pauli operators either commute or anti-commute). Consequently,
    \begin{align*}
        \tr[\rho Q \rho Q]
        =\tr[\rho Q (\tau_Q+(\rho-\tau_Q))Q]
        =\tr[\rho(\tau_Q-(\rho-\tau_Q))]
        =2\tr[\rho\tau_Q]-\tr[\rho^2]
        =\frac{1}{2^{n-1}}(1+\sum_{\one\neq P\in C(Q)}\alpha^2_P)-1\; ,
    \end{align*}
    where we used $Q^2=\one$ for every $Q\in\aP_n$ in the second step, together with $\tr[\rho^2]=1$ for pure states $\rho=\dyad{\psi}$, $|\psi\rangle\in\CP^{2^n-1}$, in the last. Finally, since $\tr[\rho a \rho b]=\tr[\rho a]\tr[\rho b]$ for pure states, we also have $\tr[(\rho Q)^2]=\tr[\rho Q \rho Q]=(\tr[\rho Q])^2=\alpha^2_Q$. Re-arranging thus yields the desired identity.
\end{proof}

Lm.~\ref{lm: Pauli identity} is closely similar to Lm.~1 in Ref.~\cite{FrembsHoehnNataleWever2026a} and suggests that purity invariants arise from sets $\cS\subset\aP_n$ whose commutants $C_\cS(P)$ are independent of the choice of $P\neq\one$. While complete independence of the commutants $C_\cS(P)$ in Lm.~\ref{lm: Pauli identity} would require $\cS=\{\one\},\aP_n,\aP^\circ_n$, the next strongest condition one can impose is that the commutants of elements inside (respectively outside) of $\cS$ intersect with equally many elements outside (respectively inside) of $\cS$, that is, $N_1(Q):=|C_\oS(Q)|=|C(Q)\cap\oS|=:N_1$ for all (hence, independent of) $Q\in\cS$ and $N_2(Q):=|C_\cS(Q)|=|C(Q)\cap\cS|=:N_2$ for all (hence, independent of) $Q\in\oS$. Here, given a subset $\cS\subset\aP^\circ_n$, we define its complement by $\oS=\aP^\circ_n\backslash\cS$. For purely illustrative purposes, an example of how a bi-regular cut could look like is shown in Fig.~\ref{fig:frustration_graph}, where each element in $\cS$ (respectively, $\oS$) commutes with $N_1=3$ ($N_2=4$) elements in $\oS$ (respectively, $\cS$).

\begin{figure}[htbp]
    \centering
    \includegraphics[width=0.4\textwidth]{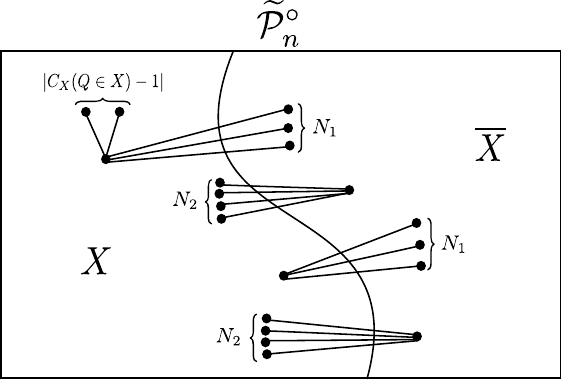} 
    \caption{Sketch of a bi-regular cut of $\aP^\circ_n$ with $N_1=3, N_2=4$ and $|C_\cS(Q\in\cS)-1|=2$. Vertices connected by straight lines commute with each other. (Not all vertices and edges included).}
    \label{fig:frustration_graph}
\end{figure}

Note first that the constants $N_1,N_2$ are not independent.

\begin{lemma}\label{lm: N1 and N2}
    Let $\cS\subset\aP^\circ_n$ and let $N_1(Q)=N_1$ for all $Q\in\cS$ and $N_2(Q)=N_2$ for all $Q\in\oS=\aP^\circ_n\backslash\cS$. Then
    \begin{align*}
        N_1
        &=(|C(Q)|-1)-|C_\cS(Q)|=(\frac{1}{2}4^n-1)-|C_\cS(Q)|\quad\forall Q\in\cS &
        N_2&=\frac{|\cS|}{|\oS|}N_1=\frac{|\cS|}{|\aP_n|-1-|\cS|}N_1\; .
    \end{align*}
\end{lemma}

\begin{proof}
    The first equality uses that every Pauli $Q\in\aP_n$ (anti-)commutes with half of the elements in $\aP_n$, that is, $|C(Q)|=\frac{1}{2}4^n$. The second follows from the identity $N_1|\cS|=\sum_{Q\in\cS} |C_{\oS}(Q)|=\sum_{Q\in\oS}|C_\cS(Q)|=N_2|\oS|=(|\aP_n|-1-|\cS|)N_2$.
\end{proof}

Clearly, constancy of $N_1(Q)$ and $N_2(Q)$ is equivalent to $\aP^\circ_n=\cS\cupdot\oS$ being a bi-regular cut of the frustration graph of $\aP_n$ as by Def.~\ref{def: frustration graph}. Our next lemma equates certain bi-regular cuts and purity invariants.

\begin{lemma}\label{lm: nondegenerate bi-regular cuts yield purity invariants}
    Let $\aP^\circ_n=\cS\cupdot\oS$ be a bi-regular cut of the frustration graph of $\aP^\circ_n$. Then 
    \begin{align}\label{eq: identity from balancedness}
        (N_2+N_1-2^{n-1}(2^n-1)+1)I_\psi(\cS)
        &=(2^n-1)N_2-(2^{n-1}-1)|\cS|\; .
    \end{align}
    In particular, $\cS$ defines a purity invariant whenever $N_2+N_1\neq 2^{n-1}(2^n-1)-1$.
\end{lemma}

\begin{proof}
    We sum both sides of Eq.~(\ref{eq: Pauli identity}) over $Q\in\cS$,
    \begin{align*}
        (2^{n-1}-1)(I_\psi(\cS) + |\cS|) &=
        (2^{n-1}-1)\sum_{Q\in\cS}(\alpha^2_Q+1)\\
        &\stackrel{\text{Lm.} \ref{lm: Pauli identity}}{=} \sum_{Q\in\cS}\sum_{\one,Q \neq P\in C(Q)} \alpha^2_P \\
        &=\sum_{Q\in\cS}\left(\sum_{Q\neq P\in C_\cS(Q)} \alpha^2_P + \sum_{P\in C_{\oS}(Q)} \alpha^2_P\right) \\
        &= (|C_\cS(Q\in\cS)| - 1)I_\psi(\cS) + N_2(I_\psi(\aP^\circ_n) - I_\psi(\cS))\; .
    \end{align*}
 Using Lm.~\ref{lm: N1 and N2}, $I_\psi(\aP^\circ_n)=2^n-1$ and  $|C(Q)|=|C_{\aP^\circ_n}(Q)|=\frac{1}{2}4^n-1$ yields the result.
\end{proof}

For bi-regular cuts with $N:=N_2+N_1-2^{n-1}(2^n-1)+1=0$ the argument in Lm.~\ref{lm: nondegenerate bi-regular cuts yield purity invariants} does not apply. In this case, $N_2=\frac{(2^{n-1}-1)}{(2^n-1)}|\cS|$ by Lm.~\ref{lm: N1 and N2}, hence, Eq.~(\ref{eq: identity from balancedness}) trivialises. Such cuts are therefore not forbidden by Lm. \ref{lm: nondegenerate bi-regular cuts yield purity invariants} and indeed exist. For instance, a single maximal Abelian subgroup $A<\aP_n$ (and its complement) defines a bi-regular cut, but clearly not a purity invariant. More generally, it can be shown (e.g. with Lm.~\ref{lm: commutants in partitions} below) that non-maximal and unextendible mutually unbiased bases lie within this family.

\section{Proof of Thm.~\ref{thm: complementary Abelian subgroups}}\label{app: proof-two observables}
For non-maximal Abelian subgroups, it is no longer sufficient to consider merely stabiliser states. We thus adapt Eq.~(\ref{eq: Pauli representation}) in order to express distributions over outcome sets of observables associated with Abelian subgroups in terms of their expectations. To this end, we label elements $\ta\in\tA\in\cA_m(\aP_n)$ using a choice of generating set $\cG=\{P_1,\cdots,P_m\}\subset\aP_n$ by $y\in\zz^m_2$ via $\ta(y)=\prod_{k=1}^mP^{y_k}_k$, and associate observables with $\tA$ via $\cG$ with outcome sets $\cX^\tA\cong\zz^m_2$. In terms of the projectors $\Pi^\tA_x$ onto the corresponding eigenspaces, these are given, for all $x\in\cX^\tA$, by
\begin{align}\label{eq: Pauli representation'}
    \Pi^\tA_x
    =\prod_{k=1}^m\frac{\one+(-1)^{x_k}P_k}{2}
    =\frac{1}{2^m}\sum_{y\in\zz^m_2}(-1)^{x\cdot y}\ta(y)\; .
\end{align}
Every state $\rho\in\cD(\C^{2^n})$ defines a probability distribution over $\cX^\tA$ (of observables associated with $\tA<\aP_n$ via $\cG$) by $p^\tA_\rho(x)=\tr[\Pi^\tA_x\rho]$ for all $x\in\cX^\tA=\zz^m_2$, where again $\Pi^\tA_x$ denotes the eigenspace corresponding to the outcome $x\in\cX^\tA$. Using Eq.~(\ref{eq: Pauli representation'}), $p^\tA_\rho$ can be expressed in terms of the expectations of elements in $\tA$ as follows,
\begin{align}\label{eq: inverse Walsh-Hadamard transform}
    p^\tA_\rho(x)
    =\tr[\Pi^\tA_x\rho]
    =\frac{1}{2^m}\sum_{y\in\zz^m_2}(-1)^{x\cdot y}\tr[\ta(y)\rho]
\end{align}
In turn, we can write the individual expectations $\tr[\ta(y)\rho]$ in terms of the distribution $p^\tA_\rho$. To see this, apply the discrete Fourier transform to Eq.~(\ref{eq: Pauli representation'}), and use character orthogonality by which
\begin{align*}
    \sum_{x\in\zz^m_2}(-1)^{x\cdot y}\Pi^\tA_x
    =\frac{1}{2^m}\sum_{y'\in\zz^m_2}\ta(y')\sum_{x\in\zz^m_2}(-1)^{x\cdot (y+y')}
    =\ta(y)\; ,
\end{align*}
and consequently,
\begin{align}\label{eq: Walsh-Hadamard transform}
    \tr[\ta(y)\rho]
    =\sum_{x\in\zz^m_2}(-1)^{x\cdot y}\tr[\Pi^\tA_x\rho]
    =\sum_{x\in\zz^m_2}(-1)^{x\cdot y}p^\tA_\rho(x)
    =\mathbb{E}_\rho[(-1)^{x\cdot y}]\; .
\end{align}
In other words, the function of Pauli expectations $y\mapsto\tr[\ta(y)\rho]$ is the (Fourier) characteristic function of the probability distribution $x\mapsto p^\tA_\rho(x)$ over its joint outcomes. With these preliminaries, we first characterise states $\rho$ of maximal information in $\tA$, that is, for which $p^\tA_\rho$ is deterministic, that is, $p^\tA_\rho=\delta_{xx_0}$.

\begin{lemma}\label{lm: max info states}
    Let $\tA<\aP_n$ be an Abelian subgroup. Then $\rho\in\cD(\C^{2^n})$ is a state of maximal information in $\tA$ if and only if it is a joint eigenstate of $\tA$, that is, if and only if there exists an (additive) character $v:\tA\ra\zz_2$ such that
    \begin{align}\label{eq: max info states}
        \tr[\ta\rho]=(-1)^{v(\ta)}\quad\quad\forall\ta\in\tA\; .
    \end{align}
\end{lemma}

\begin{proof}
    If $p^\tA_\rho(x)=\delta_{xx_0}$, then with Eq.~(\ref{eq: Walsh-Hadamard transform}),
    \begin{align*}
        \tr[\ta(y)\rho]
        =\sum_{x\in\zz^m_2}(-1)^{x\cdot y}p^\tA_\rho(x)
        =\sum_{x\in\zz^m_2}(-1)^{x\cdot y}\delta_{xx_0}
        =(-1)^{x_0\cdot y}\; ,
    \end{align*}
    and thus $v(\ta(y))=x_0\cdot y$ for every $y\in\zz^m_2$. Conversely, if $\tr[\ta(y)\rho]=(-1)^{v(\ta(y))}=(-1)^{x_0\cdot y}$ for some $x_0\in\zz^m_2$, then from Eq.~(\ref{eq: inverse Walsh-Hadamard transform}),
    \begin{equation*}
        p^\tA_\rho(x)
        =\frac{1}{2^m}\sum_{y\in\zz^m_2}(-1)^{x\cdot y}\tr[\ta(y)\rho]
        =\frac{1}{2^m}\sum_{y\in\zz^m_2}(-1)^{(x+x_0)\cdot y}
        =\delta_{xx_0}\; .\qedhere
    \end{equation*}
\end{proof}

Clearly, stabiliser states $|\psi^A\rangle$ with $\tA<A(\in\cA_n(\aP_n))$ are states of maximal information in $\tA$. Conversely, we next characterise states $\rho$ that are unbiased in $\tB\in\cA_m(\aP_n)$, that is, for which $p^\tB_\rho$ is the uniform distribution on $\cX^\tB=\zz^m_2$.

\begin{lemma}\label{lm: no info states}
    A state $\rho\in\cD(\C^{2^n})$ is unbiased in an Abelian subgroup $\tB\in\cA_m(\aP_n)$ for $m\in[n]$ if and only if
    \begin{align*}
        \tr[\tb\rho]
        =0\quad\quad\forall\one\neq\tb\in\tB\; .
    \end{align*}
\end{lemma}

\begin{proof}
    Pick a generating set $\cG=\{P_1,\cdots,P_m\}\subset\aP_n$ for $\tB\in\cA_m(\aP_n)$, that is, $\tB=\langle\cG\rangle$, and write $\tb(y)=\prod_{k=1}^mP^{y_k}_k$ for every $y\in\zz^m_2$. Now, if $\tr[\tb\rho]=0$ for all $\one\neq\tb\in\tB$ (and, clearly, $\tr[\one\rho]=1$), we immediately obtain from Eq.~(\ref{eq: inverse Walsh-Hadamard transform}),
    \begin{align*}
        p^\tB_\rho(x)
        =\tr[\Pi^\tB_x\rho]
        =\frac{1}{2^m}\sum_{y\in\zz^m_2}(-1)^{x\cdot y}\tr[\tb(y)\rho]
        =\frac{1}{2^m}\; .
    \end{align*}
    Conversely, assume that $p^\tB_\rho(x)=\frac{1}{2^m}$ for all $x\in\zz^m_2$. Then using Eq.~(\ref{eq: Walsh-Hadamard transform}) and character orthogonality yields
    \begin{equation*}
        \tr[\tb(y)\rho]
        =\sum_{x\in\zz^m_2}(-1)^{x\cdot y}p^\tB_\rho(x)
        =\frac{1}{2^m}\sum_{x\in\zz^m_2}(-1)^{x\cdot y}
        =\begin{cases}
            1&\mathrm{if}\ y=0 \\
            0&\mathrm{if}\ y\neq 0
        \end{cases}\; .\qedhere
    \end{equation*}
\end{proof}

Using Lm.~\ref{lm: max info states} and Lm.~\ref{lm: no info states}, we now generalise Cor.~\ref{cor: maximal complementary Abelian subgroups} to prove our first main result.

\begin{proof}[Proof of Thm.~\ref{thm: complementary Abelian subgroups}.]
    The only if part follows from the proof of Prop.~\ref{prop: complementary measurments}. More precisely, assume that $\tA$ and $\tB$ are complementary. By Lm.~\ref{lm: max info states}, states of maximal information in $\tA$ are joint eigenstates of $\tA$. This includes, in particular, stabiliser states $|\psi^A\rangle$ with stabiliser group $A\in\cA_n(\aP_n)$ such that $\tA\subset A$. Since $\tA$ and $\tB$ are complementary, this implies that $p^\tB_{\psi^A}$ is uniform for all $A\in\cA_n(\aP_n)$ (with $\tA\subset A$). Now, if there exists $A\supset\tA$, $A\in\cA_n(\aP_n)$ with $A\cap\tB\neq\{\one\}$, then $\tr[\tb\rho]=(-1)^{v(\tb)}\neq 0$ for all $\tb\in A\cap\tB$, hence, $p^\tB_{\psi^A}$ is not uniform by Lm.~\ref{lm: no info states}. Consequently, complementarity implies $A\cap\tB=\{\one\}$ for all $A\supset\tA$, $A\in\cA_n(\aP_n)$, equivalently $\tB\cap C(\tA)=\tB\cap A=\{\one\}$ (and similarly, for $\tA\leftrightarrow\tB$).

    Conversely, let $\rho$ be a state of maximal information in $\tA$. By Lm.~\ref{lm: max info states}, this implies $\tr[\rho\ta]=(-1)^{v(\ta)}$ for $v:\tA\ra\zz_2$ a character. Moreover, assume that $C(\tA)\cap\tB=\{\one\}$ (and similarly for $\tA\cap C(\tB)=\{\one\}$), hence, for every $\one\neq\tb\in\tB$ there exists $\ta\in\tA$ such that $[\ta,\tb]\neq 0$ and thus 
    \begin{align}\label{eq: unbiasedness from anti-commutation}
        \tr[\tb\rho]
        =\tr[\tb(\ta\rho\ta)]
        =\tr[(\ta\tb\ta)\rho]
        =-\tr[\tb\rho]=0\; .
    \end{align}
    By Lm.~\ref{lm: no info states}, this implies $p^\tB_\rho(x)=\frac{1}{2^m}$ (and similarly, for $\tA\leftrightarrow\tB$).
\end{proof}

\section{Proof of Lm.~\ref{lm: maximality bound}, Thm.~\ref{thm: max complementary sets} and Cor.~\ref{cor: existence of non-trivial maximal complementary sets}}\label{app: proof-max complementary sets}
We first establish some properties of partitions and  sets of complementary subgroups with $2^n+1$ elements before establishing (later on) that this is indeed the maximal cardinality for such sets, (and that such sets even exist in the first place). In the following, we will write $\cM=\{A_k\}_{k=0}^{2^n}$ for maximal complementary sets of Abelian subgroups (of fixed dimension $1\leq m\leq n$), and $\cS:=\cS(\cM)=\bigcupdot_{k=0}^{2^n}A^\circ_k$ for the non-identity elements contained in it.

\begin{lemma}\label{lm: commutants in partitions}
    Let $\cM=\{A_k\}_{k=0}^{2^n}\subset\cA_n(\aP_n)$ be a partition of $\aP^\circ_n=\bigcupdot_{k=0}^{2^n}A^\circ_k$ and let $B_0\subset A_0\in\cM$ be a subgroup of dimension $l$. Then for any other maximal Abelian subgroup $A_k\in\cM$, $0\neq k\leq 2^n$ there exists a subgroup $A'_k(B_0)\subset A_k$ of dimension $n-l$ that commutes with $B_0$, that is, $[A'_k(B_0),B_0]=0$.

    In particular, every $(\one\neq)Q\in\cS=\bigcupdot_{k=0}^{2^n}A^\circ_k$ commutes with $2^n(2^{n-1}-1)+2^n-1$ elements in $\cS$.
\end{lemma}

\begin{proof}
    By Lm.~\ref{lm: commutativty maps}, $|C(B_0)|=\frac{4^n}{2^l}=2^n2^{n-l}$. Clearly, the $2^n-1$ non-identity elements in $A_0$ commute with $B_0$, and in any other maximal Abelian subgroup $A_k\in\cM$, $0\neq k\leq 2^n$ there can be at most $2^{n-l}-1$ (non-identity) elements that commute with $B_0$, for otherwise there would exist an Abelian subgroup of dimension greater than $n$ in $\cP_n$. This yields an upper bound of $2^n+2^n(2^{n-l}-1)=2^n2^{n-l}$ for the number of elements commuting with $B_0$. By comparison, it follows that $B_0$ must commute with exactly $2^{n-l}$ elements in every subgroup $A_k$, $0\neq k\leq 2^n+1$
    
    Finally, since $A'_k(B_0)=C(B_0)\cap A_k$, and both $A_k,C(B_0)$ (by Lm.~\ref{lm: commutants are product closed}) are closed under products of commuting elements, it follows that $A'_k(B_0)$ is an Abelian group.
\end{proof}

We show that for $l=1$ Lm.~\ref{lm: commutants in partitions} carries over to complementary sets of Abelian subgroups of dimension $1\leq m\leq n-1$.

\begin{lemma}\label{lm: inside question generalised}
    Let $\tM=\{\tA_k\}_{k=0}^{2^n}$ be a maximal set of complementary Abelian subgroups of dimension $m$, and let $Q\in\tA_0$. For any other Abelian subgroup $\tA_k\in\tM$, $0\neq k\leq 2^n$ there exists a subgroup $\tA'_k(Q)\subset\tA_k$ of dimension $m-1$ with $[\tA'_k(Q),Q]=0$. Consequently, every $(\one\neq)Q\in\cS=\bigcupdot_{k=0}^{2^n}A^\circ_k$ commutes with $2^n(2^{m-1}-1)+2^m=2^m(2^{n-1}+1)-2^n-1$ elements in $\cS$.
\end{lemma}

\begin{proof}
    $\one\neq Q\in\tA_0$ commutes with all $2^m-1$ non-identity elements in $\tA_0\in\tM$. Moreover, for every $0\neq k\leq 2^n$, by Lm.~\ref{lm: extension of two complementry subgroups} and Lm.~\ref{lm: pairwise extendability}, we can find a partition $\cM(k)$ of $\aP^\circ_n$ containing maximal Abelian subgroups $A^\circ_0,A^\circ_k$ such that $\tA_0<A_0$ and $\tA_k<A_k$. By Lm.~\ref{lm: commutants in partitions} (with $l=1$), $Q$ commutes with a subgroup $A'_k(Q):=A'_k(\langle Q\rangle)\subset A_k$ of dimension $n-1$, hence, the subgroup $\tA'_k(Q):=A'_k(Q)\cap\tA_k$ is of dimension $m-1$ or $m$. Indeed, it cannot have dimension $m$, since otherwise $\tA_k\subset A'_k(Q)\subset C(Q)$ and thus $Q\in \tA_0\cap C(\tA_k)$, in contradiction with Thm.~\ref{thm: complementary Abelian subgroups} (since $\tA_0,\tA_k$ are complementary). Consequently, $Q$ commutes with $2^{m-1}-1$ non-identity elements in $\tA_k$ for every $0 \neq k\leq 2^n$, and thus with $2^m+2^n(2^{m-1}-1)=2^m(2^{n-1}+1)-2^n-1$ elements in $\cS$.
    
    Finally, since $\tA'_k(Q)=C(Q)\cap\tA_k$, and both $\tA_k,C(Q)$ (by Lm.~\ref{lm: commutants are product closed}) are closed under products of commuting elements, $\tA'_k(Q)\subset\tA_k$ is an Abelian subgroup.
\end{proof}

Lm.~\ref{lm: inside question generalised} counts the number of elements in $\cS$ that commute with any given element $Q\in\cS$. Before counting the number of elements in $\cS$ commuting with an element $Q\notin\cS$ in its complement, we first prove the following lemma.

\begin{lemma}\label{lm: balanced intersection}
    Let $\tM=\{\tA_i\}_{k=0}^{2^n}\subset\cA_m(\aP_n)$ be a maximal set of complementary subgroups of dimension $m$. Then $\cS=\bigcupdot_{k=0}^{2^n}\tA^\circ_k$ intersects every maximal Abelian subgroup in exactly $2^m-1$ elements.
\end{lemma}

\begin{proof}
    For every maximal Abelian subgroup $A\in\cA_n:=\cA_n(\aP_n)$, define $c_A=|\cS\cap A^\circ|$. We will determine $c_A$ from the following two ways of counting the incidences of elements of $\cS$ with elements of $\cA_n$:
    \begin{align}\label{eq: average single intersection}
        \sum_{A\in\cA_n}c_A
        =\sum_{P\in\cS}\Omega_1(P)
        =|\cS|\cdot\Omega_1\; ,
    \end{align}
    where $\Omega_1=\Omega_1(P):=|\{A\in\cA_n\mid P\in A\}|$ for every $\one\neq P\in\aP_n$ is in fact constant on $\aP^\circ_n$ by symmetry of the Pauli group, in particular, of its (maximal) Abelian subgroups. More precisely, since any two non-identity Pauli operators $P,Q\in\aP^\circ_n$ are related by a Clifford unitary $U\in\Cl_n$, that is, $Q=UPU^\dagger$, and $U\cA_n U^\dagger=\cA_n$, one has
    \begin{equation}\label{eq: Clifford invariance}
    \begin{aligned}
        \Omega_1(Q)
        =\Omega_1(UPU^\dagger)
        &=|\{A\in\cA_n\mid UPU^\dagger\in A\}|\\
        &=|\{U^\dagger AU\in\cA_n\mid P\in U^\dagger AU\}|\\
        &=|\{A'\in U\cA_n U^\dagger\mid P\in A'\}|
        =|\{A'\in\cA_n\mid P\in A'\}|
        =\Omega_1(P)\; .
    \end{aligned}
    \end{equation}
    Similarly, we count the number of incidences of two distinct commuting elements in $\cS$ with maximal Abelian subgroups of $\aP_n$, that is, with elements in $\cA_n$, as follows,
    \begin{align}\label{eq: average pair intersection}
        \sum_{A\in\cA_n}c_A(c_A-1)
        =\sum_{\substack{P,Q\in\cS,[P,P']=0}}\Omega_2(P,Q)
        =|\cS|\cdot|(C_\cS(Q\in\cS)|-1)\cdot\Omega_2\; ,
    \end{align}
    where $\Omega_2=\Omega_2(P,Q):=|\{A\in\cA_n\mid P,Q\in A\}|$ for any commuting pair $\one\neq P\neq Q\neq\one$, $[P,Q]=0$, is again a constant (by a similar argument to Eq.~(\ref{eq: Clifford invariance})), and we use that $|C_\cS(Q\in\cS)|=2^m+2^n(2^{m-1}-1)-1$ is a constant by Lm.~\ref{lm: inside question generalised}. (Here, we exclude $Q$ from the counting since $\one\notin\cS$ and the left hand side counts distinct pairs of commuting elements $P,Q$.) To evaluate these expressions, we recall the following known counts,
    \begin{align}\label{eq: stabiliser counts}
        |\cA_n|
        &=\prod_{i=1}^n(2^i+1) &
        \Omega_1
        &=\prod_{i=1}^{n-1}(2^i+1) &
        \Omega_2
        &=\prod_{i=1}^{n-2}(2^i+1)\; .
    \end{align}
    In fact, recall that the number of $n$-qubit stabiliser states stands at $2^n\prod_{i=1}^n(2^i+1)$ (see e.g. Prop.~2 in Ref.~\cite{AaronsonGottesman2004}), from which the count for maximal Abelian subgroups $\Omega_0:=|\cA_n|$ follows by quotienting out by the group of characters $(-1)^{x\cdot y}$ for $x\in\zz^n_2$ on $A$ (since a stabiliser state $|\psi^{A,x}\rangle$ is uniquely determined by the pair $(A,x)$, see Eq.~(\ref{eq: Pauli representation})). The other counts follow similarly from Lm.~\ref{lm: commutant MUB sub-partition}: note that the number of maximal Abelian subgroups containing a single element $P\in\aP^\circ_n$, or a pair of distinct (more generally, a set of independent and) commuting elements $P,Q\in\aP^\circ_n$, $[P,Q]=0$ is precisely the number of maximal Abelian subgroups in the commutants $C(P)$, respectively $C(P,Q)=C(\langle P,Q\rangle)$; by the proof of Lm.~\ref{lm: commutant MUB sub-partition}, (the quotients of) these subgroups are in bijective correspondence with the maximal abelian subgroups of $\aP_m$ for $m=n-1,n-2$, respectively, and the result thus follows by the counting for $\cA_n$. Finally, from plugging Eq.~(\ref{eq: stabiliser counts}), together with $|\cS|=(2^n+1)(2^m-1)$ and $|C_\cS(Q\in\cS)|-1=(2^m-2)(2^{n-1}+1)$ into Eq.~(\ref{eq: average pair intersection}), we get
    \begin{align*}
        \sum_{A\in\cA_n}c_A(c_A-1)
        =(2^n+1)(2^m-1)(2^m-2)(2^{n-1}+1)\prod_{i=1}^{n-2}(2^i+1)
        =(2^m-1)(2^m-2)|\cA_n|\; .
    \end{align*}
    The two constraints force the value of $c_A$ to be constant, indeed, let $c=2^m-1$, then
    \begin{align*}
        \sum_{A\in\cA_n}(c_A-c)^2
        &=\sum_{A\in\cA_n}c^2_A-2c\sum_{A\in\cA_n}c_A+c^2\sum_{A\in\cA_n} \\
        &=\sum_{A\in\cA_n}c_A(c_A-1)-(2c-1)\sum_{A\in\cA_n}c_A+c^2\sum_{A\in\cA_n} \\
        &=c(c-1)|\cA_n|-(2c-1)c|\cA_n|+c^2|\cA_n|
        =0\; .\qedhere
    \end{align*}
\end{proof}

From this, we immediately obtain that every element not contained in a set of $2^n+1$ complementary Abelian subgroups commutes with the same number of elements within this set.

\begin{lemma}\label{lm: outside question generalised}
    Let $\tM=\{\tA_k\}_{k=0}^{2^n}\subset\cA_m(\aP_n)$ be a maximal set of complementary Abelian subgroups of dimension $m$. Then every $(\one\neq)Q\notin\cS=\bigcupdot_{k=0}^{2^n}\tA^\circ_k$ commutes with $(2^m-1)(2^{n-1}+1)$ elements in $\cS$.
\end{lemma}

\begin{proof}
    Lm.~\ref{lm: commutant MUB sub-partition} with $m=1$ asserts that the elements commuting with $\one\neq Q\in\oS$ admit a partition into $2^{n-1}+1$ maximally commuting subgroups of cardinality $2^{n-1}$, which correspond with $2^{n-1}+1$ maximal Abelian subgroups including $Q$ (and closed under multiplication). By Lm.~\ref{lm: balanced intersection}, $\cS$ contains exactly $2^m-1$ elements from each subgroup, hence, $|C_\cS(Q)|=(2^{n-1}+1)(2^m-1)$ for every $\one\neq Q\in\oS$.
\end{proof}

Using the above properties of sets of $2^n+1$ complementary Abelian subgroups, we prove the bound on maximal sets of complementary Abelian subgroups of dimension $1\leq m\leq n$.

\begin{proof}[Proof of Lm.~\ref{lm: maximality bound}]
    Assume otherwise that $K>2^n+1$. Pick the first $2^n+1$ Abelian subgroups from the set of $K$ elements, and let $\cS=\bigcupdot_{k=1}^{2^n+1}\tA^\circ_k$ be the union of all non-identity elements. Since the condition of complementarity applies pairwise, the subset $\{\tA_k\}_{k=1}^{2^n+1}$ is still complementary, and the above arguments thus apply. In particular, by Lm.~\ref{lm: outside question generalised}, we have $|C_\cS(Q)|=(2^m-1)(2^{n-1}+1)$ for every $\one\neq Q\notin\cS$. Let $\one\neq Q\in\tA_l$ for $l>2^n+1$. Then by the same reasoning as in Lm.~\ref{lm: inside question generalised}, applied to the pairs $(\tA_k,\tA_l)$ for $l\leq 2^n+1$, the subgroups $C(Q)\cap\tA_i$ are all of dimension $m-1$. Summing over these then results in the incompatible (for any $1\leq m\leq n$) alternative count $|C_\cS(Q)|=(2^n+1)(2^{m-1}-1)$.
\end{proof}

Moreover, we are now in the position to prove our second main result.

\begin{proof}[Proof of Thm.~\ref{thm: max complementary sets}]
    By Lm.~\ref{lm: inside question generalised} and Lm.~\ref{lm: outside question generalised}, the commutants of elements $Q\in\cS$ (respectively, $\one\neq Q\in\oS=\aP_n\backslash\cS$) intersect with equally many elements in a maximal set of complementary subgroups $\tM$, that is, $|C_\oS(Q)|=|C(Q)\cap\oS|=\mathrm{const}$ for all (hence, independent of) $Q\in\cS$ and $|C_\cS(Q)|=|C(Q)\cap\cS|=\mathrm{const}$ for all (hence, independent of) $Q\in\oS$. Consequently, $(\cS,\oS^\circ)$ defines a bi-regular cut of the frustration graph of $\aP^\circ_n$ (see Def.~\ref{def: frustration graph} and Def.~\ref{def: bi-regular cuts} in App.~\ref{app: Pauli group}). Moreover, from Lm.~\ref{lm: inside question generalised} and Lm.~\ref{lm: outside question generalised} again, we find that $N_2=|C_\cS(Q\in\oS)|=(2^m-1)(2^{n-1}+1)$ and $N_1=(\frac{1}{2}4^n-1)-|C_\cS(Q\in\cS)|=\frac{1}{2}4^n-2^m(2^{n-1}+1)+2^n$ by Lm.~\ref{lm: N1 and N2}, and thus $N_1+N_2=2^{n-1}(2^n+1)-1\neq 2^{n-1}(2^n-1)-1$. By Lm.~\ref{lm: nondegenerate bi-regular cuts yield purity invariants}, $\cS$ thus defines a purity invariant with $I_\psi(\cS)=2^m-1$ for every $\psi\in\CP^{2^n-1}$.
\end{proof}

\begin{proof}[Proof of Cor.~\ref{cor: existence of non-trivial maximal complementary sets}]
    The purity invariants $J$ in Ref.~\cite{FrembsHoehnNataleWever2026a} satisfy $I_\psi(J)=2^{n-1}$ and are closed under anti-commutators (see Lm.~3 in Ref.~\cite{FrembsHoehnNataleWever2026a}). Consequently, their intersection with any partition $\{A^\circ_i\}_{k=1}^{2^n+1}$ of $\aP^\circ_n$ into maximal Abelian subgroups $A_k\in\cA_n$, defines a set of Abelian subgroups $\{\tA_k\}_{k=1}^{2^n+1}$ with $\tA_k=J\cap A_k$ of dimension $n-1$. Moreover, any two different subgroups in this set satisfy the condition in Thm.~\ref{thm: complementary Abelian subgroups} by means of being a purity invariant.
\end{proof}

In fact, the last argument proves that any purity invariant that is closed under anti-commutators, and thus $I_\psi(\cS)=2^m$ for $1\leq m\leq n$, gives rise to a maximal set of complementary Abelian subgroups of dimension $m$. In Ref.~\cite{Frembs2026}, we show that the family in Ref.~\cite{FrembsHoehnNataleWever2026a} with $m=n-1$ are the only anti-commutator closed purity invariants.

\section{Brukner-Zeilinger information measure and purity invariants}\label{app: purity invariants}

Def.~\ref{def: complementarity} involves the comparison of states of maximal and minimal uncertainty (equivalently, minimal and maximal information). As such, the precise choice of quantitative measure of entropy (equivalently, information) is irrelevant as long as different measures agree on minimal and maximal uncertainty (respectively, information) states. For comparison with our companion papers \cite{FrembsHoehnNataleWever2026a,Frembs2026} and in preparation of Thm.~\ref{thm: collision entropy bound} below, we recall the \emph{Brukner-Zeilinger information measure} \cite{BruknerZeilinger1999,BruknerZeilinger2001,Hoehn2017,HoehnWever2017}, which for a set of Abelian subgroups $\{A_k\}_{k=1}^K\subset\cA_m(\aP_n)$ is defined by
\begin{align}\label{eq: BZ info measure - subgroups}
    I_\rho(\{A_k\}_{k=1}^K)
    :=1+\sum_{k=1}^K\sum_{\one\neq P\in A_k}\tr[P\rho]^2\; .
\end{align}
As we have seen, for (observables associated with) Abelian subgroups of the $n$-qubit Pauli group, complementarity reduces to intersection properties of the corresponding Abelian subgroups. For instance, the information of a stabiliser state $|\psi^A\rangle$ with stabiliser group $A\in\cA_n(\aP_n)$ of an (observable associated with the) Abelian subgroup $B\in\cA_m(\aP_n)$ is a function of the overlap $|A\cap B|$; in fact, $\log_2(|A\cap B|)$ represents the number of bits required to specify any eigenstate of the overlap, and thus the information about $B$ contained in $|\psi^A\rangle$. With this in mind, the following lemma establishes Eq.~(\ref{eq: BZ info measure - subgroups}) as a natural information measure for the study of complementarity in the $n$-qubit Pauli group.

\begin{lemma}\label{lm: quantised stab information}
    Let $|\psi^A\rangle$ be a stabiliser state with stabiliser group $A\in\cA_n(\aP_n)$, and let $B\in\cA_n(\aP_n)$ be a maximal Abelian subgroup. Then $I_{|\psi^A\rangle}(B)=|A\cap B|$.
\end{lemma}

\begin{proof}
    Let $|A \cap B|=2^m$ for $0\leq m\leq n$. Then there exist generating sets $\cG_A=\{P_1,\cdots,P_n\}$ and $\cG_B=\{Q_1,\cdots,Q_n\}$ for $A$ and $B$ with $P_i=Q_i$ for $1\leq i \leq m$. Consequently, the elements in the joint eigenbasis $\{|\psi^{B,x}\rangle\}_{x\in\zz_2^n}$ of $B$ satisfy $P_i|\psi^{B,x}\rangle=Q_i |\psi^{B,x}\rangle=(-1)^{v_x(Q_i)}|\psi^{B,x}\rangle$ for all $1 \leq i\leq m$ and $x\in\zz_2^n$ and character $v_x:B\ra\zz_2$, and therefore
    \begin{align*}
        I_{|\psi^A\rangle}(B)
        &= \sum_{Q\in A\cap B} I_{|\psi^A\rangle}(Q) + \sum_{Q\in B\backslash (A\cap B)} I_{|\psi^A\rangle}(Q) \\
        &=\sum_{y\in\zz_2^m \times O^{n-m}} |\langle\psi^A|b(y)|\psi^A\rangle|^2+\sum_{y\notin \zz_2^m\times O^{n-m}} |\langle\psi^A|\left(\sum_{x\in\zz_2^n} (-1)^{\sum_{i=1}^n x_iy_i} \dyad{\psi^{B,x}}\right)|\psi^A\rangle|^2\\
        &=\sum_{y\in\zz_2^m \times O^{n-m}}|\langle\psi^A|\psi^A\rangle|^2 + \sum_{y\notin\zz_2^m \times O^{n-m}} \frac{1}{2^{2(n-m)}}\left(\sum_{\substack{x\in\zz_2^n\\|\langle\psi^A|\psi^{B,x}\rangle|^2\neq0}}(-1)^{\sum_{k=1}^n x_ky_k} \right)^2\\
        &=|A\cap B|\; ,
    \end{align*}
    where $O^{n-m}$ is the zero vector of length $n-m$ and we used Eq.~(\ref{eq: Pauli representation}) in the first step, Lm.~\ref{lm: AG} by which either $|\langle\psi^A|\psi^{B,x}\rangle|^2=0$ or $|\langle\psi^A|\psi^{B,x}\rangle|^2=\frac{1}{2^{n-m}}$ (since $|A \cap B|=2^m$) in the third, and orthogonality of characters, explicitly that the squared sum in the second last line runs over an equal number of positive and negative terms in the last step.
\end{proof}

More generally, in Ref.~\cite{Frembs2026,FrembsHoehnNataleWever2026a} we study all subsets $\cS\subset\aP_n$ which define purity invariants:
\begin{align}\label{eq: purity invariants}
    I_\psi(\cS)
    =2^{n-1}\quad\quad\forall\psi\in\CP^{2^n-1}\; .
\end{align}
For two qubits, these invariants were first discovered and played a crucial part in the reconstruction programme in Ref.~\cite{Hoehn2017,HoehnWever2017}. Only in this case ($n=2$), do they correspond with maximal sets of mutually anti-commuting Pauli operators, which in turn are known to give rise to strong entropic uncertainty relations \cite{WehnerWinter2008}. However, as evident from Thm.~\ref{thm: max complementary sets} the property of mutual anti-commutation does not generalise to maximal complementary sets beyond the two-qubit case (for details, see Ref.~\cite{FrembsHoehnNataleWever2026a}), since maximal sets of mutually anti-commuting Pauli operators have cardinality $2n+1<2^n+1$ for $n>2$ \cite{SarkarVanDenBerg2021}. Nevertheless, in App.~\ref{app: entropic UR}, we use Eq.~(\ref{eq: purity invariants}) to prove that maximal complementary sets, too, give rise to strong bounds on entropic uncertainty relations.

\section{Entropic uncertainty relations}\label{app: entropic UR}

A general entropic uncertainty relation for a set of observables $\{M_k\}_{k=1}^K$ and a quantum state $\rho$ takes the form
\begin{align}\label{eq: general entropic UR}
    \frac{1}{K}\sum_{k=1}^K H(M_k\mid\rho)\
    \geq\ c(\{M_k\}_{k=1}^K)\; ,
\end{align}
where $H$ is any entropy function measuring the amount of uncertainty in the observables $M_k$ and $c(\{M_k\}_{k=1}^K)$ some constant. A common choice is the family of R\`enyi entropies \cite{Renyi1961}, which is defined for the parameter range $0<\alpha<\infty$, $\alpha\neq 1$ by
\begin{align}\label{eq: Renyi entropy}
    H_\alpha(M_k\mid p_k)
    :=\frac{1}{1-\alpha}\log\left(\sum_{x\in\cX_k}p^\alpha_k(x)\right)\; ,
\end{align}
where $p_k(x)$ is the probability distribution over the outcome set $\cX_k$ of an observable $M_k$. Moreover, the limit $\alpha\ra 1$ corresponds with the Shannon entropy $H(M_k\mid p_k)=-p_k(x)\log(p_k(x))$. Here, we will consider the collision entropy, $H_2(M_k\mid p_k)=-\log\left(\sum_{x\in\cX}p^2_k(x)\right)$, as it is closely related (as will be seen further below) with the Brukner-Zeilinger information measure in Eq.~(\ref{eq: BZ info measure}).

Crucially, the bound $c(\{M_k\}_{k=1}^K)$ in Eq.~(\ref{eq: general entropic UR}) holds for all quantum states $\rho$, hence, is independent of $\rho$ \cite{Deutsch1983}. For their use in cryptographic protocols, it is desirable to find measurement arrangements $\{M_k\}_{k=1}^K$ for which $c(\{M_k\}_{k=1}^K)$ becomes as large as possible.\footnote{In many cryptographic tasks, one further needs to account for side information \cite{ChristandlWinter2005,BertaEtAl2010}, and consider various specific entropies, e.g. in quantum key distribution \cite{TomamichelEtAl2011,TomamichelRenner2011}.} Assuming that $|\cX_k|=|\cX|$, a fundamental bound (for any normalized entropy function $H$) is given by \cite{WehnerWinter2010},
\begin{align*}
    0\ \leq\ c(\{M_k\}_{k=1}^K)\ \leq\ \frac{K-1}{K}\log(|\cX|)\; ,
\end{align*}
where we use the normalisation $0\leq H(\cdot)\leq\log(|\cX|)$. Equality in the upper bound implies that certainty about any one outcome of any one measurement $M_k$ implies maximal uncertainty about all the others $M_{k'}\neq M_k$. Conversely, since the latter is the case for complementary sets of observables, one may expect strong uncertainty bounds in these cases. However, for general (not necessarily maximal) complementary sets this is not the case: as shown in Ref.~\cite{BallesterWehner2007}, there exist (non-maximal) complementary sets of maximal Abelian subgroups whose bound is merely the cumulative bound of the pairwise uncertainty relations. By contrast, we now prove Thm.~\ref{thm: collision entropy bound}, which asserts that for maximal complementary sets of Abelian subgroups, one indeed obtains strong uncertainty bounds.

\begin{proof}[Proof of Thm.~\ref{thm: collision entropy bound}.]
    We recall from Eq.~(\ref{eq: Pauli representation'}), that we can express the projectors corresponding to outcomes $x_k\in\Xi_k=[2^m]:=\{1,\cdots,2^m\}$ of observable corresponding to $\tA_k$ as follows
    \begin{align*}
        \Pi_{k,x}
        &=\frac{1}{2^m}\sum_{y\in\zz^m_2}(-1)^{x\cdot y}\ta_k(y) &
        \tr[\Pi_{k,x}]
        &=2^{n-m}\; .
    \end{align*}
    The probability of obtaining outcome $x$ in measurement $M_k$ (associated with $\tA_k$) in the state $\rho=\frac{1}{2^n}\sum_{P\in\aP_n}\alpha_P P$ (with $\alpha_P\in\R$ and $\alpha_\one=1$) is then given by Eq.~(\ref{eq: inverse Walsh-Hadamard transform}) as $p_{k,\rho}(x)
    =\tr[\Pi_{k,x}\rho]
    =\frac{1}{2^m}\sum_{y\in\zz^m_2}(-1)^{x\cdot y}\tr[\ta_k(y)\rho]$. Using $\tr[PQ]=2^n\delta_{PQ}$ for $P,Q\in\cP_n$ and $\sum_{x\in\zz^m_2}(-1)^{x\cdot (y+y')}=2^m\delta_{yy'}$, we compute the collision probability as
    \begin{align*}
        C(p_{k,\rho})
        :=\sum_{x\in\zz^m_2}p^2_{k,\rho}(x)
        =\frac{1}{2^{2m}}\sum_{x\in\zz^m_2}\sum_{y,y'\in\zz^m_2}(-1)^{x\cdot (y+y')}\alpha_{\ta_k(y)}\alpha_{\ta_k(y')}
        =\frac{1}{2^m}\sum_{y\in\zz^m_2}\alpha^2_{\ta_k(y)}
        =\frac{1}{2^m}(1+\sum_{\one\neq P\in\tA_k}\alpha^2_P)\; .
    \end{align*}
    Using that $\{\tA_k\}_{k=1}^{2^n+1}$ defines a purity invariant with $I_\psi(\{\tA_k\}_{k=1}^{2^n+1})=2^m$ by Thm.~\ref{thm: max complementary sets}, we have
    \begin{equation*}
        \sum_{k=1}^{2^n+1}C(p_{k,\psi})
        =\sum_{k=1}^{2^n+1}\sum_{x\in\zz^m_2}p^2_{k,\psi}(x)
        =\sum_{k=1}^{2^n+1}\frac{1}{2^m}(1+\sum_{\one\neq P\in\tA_k}\alpha^2_P)
        =\frac{2^n+I_\psi(\{\tA_k\}_{k=1}^{2^n+1})}{2^m}
        =2^{n-m}+1\; ,
    \end{equation*}
    for every pure state $\rho=\dyad{\psi}\in\CP^{2^n-1}$. Finally, the bound in Eq.~(\ref{eq: collision entropy bound}) follows from concavity of the logarithm,
    \begin{align*}
        \frac{1}{2^n+1}\sum_{k=1}^{2^n+1}H_2(M_k|\psi)
        &=-\frac{1}{2^n+1}\sum_{k=1}^{2^n+1}\log\left(\sum_{x\in\zz^m_2}p^2_{k,\psi}(x)\right) \\
        &\geq-\log\left(\frac{1}{2^n+1}\sum_{k=1}^{2^n+1}\sum_{x\in\zz^m_2}p^2_{k,\psi}(x)\right)
        =\log(\frac{2^n+1}{2^{n-m}+1})\; .\qedhere
    \end{align*}
\end{proof}

Using that $H_\alpha(\cdot)\geq H_{\alpha'}(\cdot)$ for all $\alpha<\alpha'$, one deduces analogous bounds also for $\alpha\leq 2$. The proof of Thm.~\ref{thm: collision entropy bound} is similar to the one for the case $m=n$ in Ref.~\cite{BallesterWehner2007}, which uses that mutually unbiased bases define complex projective $2$-designs \cite{KlappeneckerRoetteler2005}. The proof of the entropic uncertainty relation in Eq.~(\ref{eq: collision entropy bound}) instead uses the fact that maximal complementary sets define purity invariants \cite{FrembsHoehnNataleWever2026a,Frembs2026}. In fact, in Ref.~\cite{Frembs2026} we prove that the Clifford stabiliser groups of the purity invariants in Ref.~\cite{FrembsHoehnNataleWever2026a} define complex projective state $2$-designs as well.

\end{document}